\documentclass[runningheads]{llncs}

\usepackage{graphicx}
\usepackage{amsmath}
\usepackage{amssymb}
\usepackage{amsfonts}
\usepackage{mathtools}
\usepackage{tikz}
\usepackage{pgfplots}
\pgfplotsset{compat=1.16}
\usepackage{listings}
\usepackage{booktabs}
\usepackage{multirow}
\usepackage{array}
\usepackage{algorithm}
\usepackage{algpseudocode}
\usepackage{enumitem}
\usepackage{float}
\usepackage{hyperref}
\usepackage{url}
\usepackage{xcolor}

\hypersetup{
    colorlinks=true,
    linkcolor=blue,
    filecolor=magenta,
    urlcolor=cyan,
    citecolor=blue
}

\tikzset{every picture/.style={line width=0.75pt}}
\pgfplotsset{
    every axis/.append style={
        line width=0.5pt,
        tick style={line width=0.5pt}
    }
}

\begin{document}

\title{Computing Supported Models via Transformation to Stable Models}

\author{Fang Li\inst{1} \and Gopal Gupta\inst{2}}

\institute{Oklahoma Christian University \and
University of Texas at Dallas}
\maketitle

\begin{abstract}
Answer Set Programming (ASP) with stable model semantics has proven highly effective for knowledge representation and reasoning. However, the minimality requirement of stable models can be restrictive for applications requiring exploration of non-minimal but logically consistent solution spaces. Supported models, introduced by Apt, Blair, and Walker in 1988, relax this minimality constraint while maintaining a support condition ensuring every true atom is justified by some rule.

Despite their theoretical significance, supported models lack practical computational tools integrated with modern ASP solvers. We present a novel transformation-based method enabling computation of supported models using standard ASP infrastructure. Our approach transforms any ground logic program into an equivalent program whose stable models correspond exactly to the supported models of the original program. We implement this transformation for Clingo, providing a practical tool for computing supported models with state-of-the-art ASP solvers.

We demonstrate applications in software verification, medical diagnosis, and planning where supported models enable valuable exploratory reasoning capabilities beyond those provided by stable models. Our implementation is publicly available and compatible with standard ASP syntax.

\end{abstract}

\section{Introduction}

Answer Set Programming (ASP) is a powerful declarative paradigm for solving complex combinatorial search and reasoning problems \cite{baral2003knowledge}. Its foundation in stable model semantics \cite{gelfond1988stable} identifies solutions as minimal models satisfying program rules, which is effective for conservative reasoning. However, modern applications often demand more flexible reasoning capabilities \cite{faber2011answer}, where non-minimal interpretations can be valuable.

\subsection{Motivation}

The strict minimality of stable model semantics, while theoretically well-founded, can limit the discovery of potentially useful solutions. Consider a program representing possible states or outcomes; stable models capture only the most minimal scenarios, potentially missing richer, non-minimal states that are still consistent with the program's rules. This limitation is particularly relevant in domains such as \textbf{Medical Diagnosis} (where multiple co-occurring conditions might consistently explain symptoms, offering comprehensive differential diagnosis), \textbf{Planning} (where non-minimal plans could offer greater robustness via redundant steps), and \textbf{Software Verification} (where non-minimal test scenarios reflect real-world interaction combinations).

Supported models, introduced by Apt, Blair, and Walker \cite{apt1988towards}, address this limitation by relaxing the minimality requirement while preserving a crucial property: every true atom must be \emph{supported} by some rule whose body is satisfied. This ensures logical consistency while permitting non-minimal models. Marek and Subrahmanian \cite{marek1992relationship} established that every stable model is a supported model, but not conversely.

Despite their theoretical importance, supported models have seen limited practical adoption because existing computational methods rely on specialized fixpoint interpreters or established transformations to propositional logic (Clark's Completion and SAT solving). Current ASP solvers like Clingo \cite{gebser2014clingo} are highly optimized for stable models but provide no direct, integrated support for computing supported models.

\subsection{Our Contribution}

We present a transformation-based approach bridging this gap by enabling computation of supported models using standard stable model solvers. Our key contributions are:

\begin{enumerate}
    \item \textbf{Transformation method}: A polynomial-time transformation converting any ground logic program $\Pi$ into program $\Pi'$ such that stable models of $\Pi'$ correspond exactly to supported models of $\Pi$.
    
    \item \textbf{Theoretical results}: Proof of transformation correctness establishing a bijection between supported models of the original program and stable models of the transformed program.
    
    \item \textbf{Practical implementation}: Implementation for Clingo providing a practical tool for computing supported models integrated with modern ASP infrastructure.
    
    \item \textbf{Application examples}: Demonstrations in software verification, medical diagnosis, and planning where supported models enable exploration of solution spaces that stable models cannot capture.
\end{enumerate}

The remainder of this paper is organized as follows. Section 2 provides background on ASP and supported models. Section 3 presents our transformation method with correctness proofs. Section 4 describes the implementation and algorithm details. Section 5 discusses the relationship to Clark's Completion and the grounding restriction. Section 6 provides application examples. Section 7 reviews related work, and Section 8 concludes.

\section{Background}

\subsection{Answer Set Programming Fundamentals}

An ASP program $\Pi$ consists of rules of the form:
\[h \leftarrow b_1, \ldots, b_m, \textnormal{not}\; c_1, \ldots, \textnormal{not}\; c_n\]
where $h, b_i, c_j$ are atoms and $\textnormal{not}$ represents negation as failure. The set of atoms in a program $\Pi$ is denoted by $\textnormal{atoms}(\Pi)$. 

For a rule $r$, we denote:
\begin{itemize}
    \item $\textnormal{head}(r) = h$
    \item $\textnormal{body}^+(r) = \{b_1, \ldots, b_m\}$ (positive body literals)
    \item $\textnormal{body}^-(r) = \{c_1, \ldots, c_n\}$ (negated body literals)
\end{itemize}

A set of atoms $M$ is a \emph{model} of a rule if $h \in M$ whenever $\{b_1, \ldots, b_m\} \subseteq M$ and $\{c_1, \ldots, c_n\} \cap M = \emptyset$. $M$ is a model of $\Pi$ if it is a model of every rule in $\Pi$.

\subsection{Stable Model Semantics}

The stable model semantics \cite{gelfond1988stable} is the foundation of ASP. For a ground program $\Pi$ and a set of atoms $M$, the Gelfond-Lifschitz reduct $\Pi^M$ is obtained by:
\begin{enumerate}
    \item Removing from $\Pi$ each rule that has a negative literal $\textnormal{not}\; c$ in its body with $c \in M$, and
    \item Removing all negative literals $\textnormal{not}\; c$ from the bodies of the remaining rules.
\end{enumerate}

$\Pi^M$ is a positive program. A set $M$ is a \emph{stable model} of $\Pi$ if and only if $M$ is the minimal model of $\Pi^M$.

\subsection{Supported Models}

We now recall the definition of supported models from Apt, Blair, and Walker \cite{apt1988towards}.

\begin{definition}[Supported Model]\label{def:supported}
Given a ground logic program $\Pi$, a set of atoms $M$ is a \emph{supported model} of $\Pi$ if:
\begin{enumerate}
    \item $M$ is a model of $\Pi$ (i.e., $M$ satisfies all rules in $\Pi$), and
    \item For each atom $a \in M$, there exists a rule $r \in \Pi$ such that:
        \begin{itemize}
            \item $\textnormal{head}(r) = a$
            \item $\textnormal{body}^+(r) \subseteq M$
            \item $\textnormal{body}^-(r) \cap M = \emptyset$
        \end{itemize}
\end{enumerate}
We say such a rule $r$ \emph{supports} atom $a$ in $M$.
\end{definition}

For non-ground programs, this definition applies to the ground instantiation over the Herbrand universe.

\begin{theorem}[Relationship to Stable Models \cite{marek1992relationship}]\label{thm:stable-supported}
Every stable model of a logic program is a supported model. The converse does not hold in general.
\end{theorem}

\begin{example}\label{ex:supported-not-stable}
Consider the program:
\begin{align*}
p &\leftarrow q, \textnormal{not}\; r \\
q &\leftarrow p
\end{align*}

The empty set $\emptyset$ is both a stable model and a supported model.

The set $\{p, q\}$ is a supported model but not a stable model:
\begin{itemize}
    \item $\{p,q\}$ is a model (both rules satisfied)
    \item $p$ is supported by the first rule (since $q \in \{p,q\}$ and $r \notin \{p,q\}$)
    \item $q$ is supported by the second rule (since $p \in \{p,q\}$)
    \item But $\{p,q\}$ is not stable: the reduct w.r.t.\ $\{p,q\}$ has least model $\emptyset$, not $\{p,q\}$
\end{itemize}
\end{example}

\section{Transformation Method}

We now present our main contribution: a transformation enabling computation of supported models via stable model solvers.

\subsection{The Transformation}

The key insight is to decouple conjunctive rule bodies using auxiliary atoms representing ``the body is disabled.'' For a rule $r$, we introduce auxiliary atom $\_dm\_r$ (mnemonic: ``De Morgan for rule $r$'') that is true exactly when the body of $r$ is false.

\begin{definition}[Transformation]\label{def:transformation}
Let $\Pi$ be a ground logic program. The \emph{transformed program} $T(\Pi)$ is constructed as follows. For each rule $r$ in $\Pi$ of the form:
\[
h \leftarrow l_1, \ldots, l_m, \textnormal{not}\; l_{m+1}, \ldots, \textnormal{not}\; l_n
\]
we generate the following rules in $T(\Pi)$:
\begin{align}
h &\leftarrow \textnormal{not}\; \_dm\_r \label{eq:t1} \\
\_dm\_r &\leftarrow \textnormal{not}\; l_i \quad \textnormal{for each } i \in \{1, \ldots, m\} \label{eq:t2} \\
\_dm\_r &\leftarrow l_j \quad \textnormal{for each } j \in \{m+1, \ldots, n\} \label{eq:t3}
\end{align}
\end{definition}

\begin{remark}
The atom $\_dm\_r$ is true if and only if the body of rule $r$ is false. Rules \eqref{eq:t2} and \eqref{eq:t3} implement De Morgan's law:
\[
\neg(l_1 \land \cdots \land l_m \land \neg l_{m+1} \land \cdots \land \neg l_n) \equiv \neg l_1 \lor \cdots \lor \neg l_m \lor l_{m+1} \lor \cdots \lor l_n
\]
This transformation is structurally related to \textbf{Clark's Completion} in capturing the supported model property. However, our method directly leverages the stable model semantics of a modern ASP solver for computation. We discuss this relationship in detail in Section~\ref{sec:clarks}.
\end{remark}

\begin{example}\label{ex:transformation}
Consider the rule:
\[
p \leftarrow q, r, \textnormal{not}\; s
\]
The transformation produces:
\begin{align*}
p &\leftarrow \textnormal{not}\; \_dm\_r \\
\_dm\_r &\leftarrow \textnormal{not}\; q \\
\_dm\_r &\leftarrow \textnormal{not}\; r \\
\_dm\_r &\leftarrow s
\end{align*}
\end{example}

\subsection{Correctness}

We now prove that our transformation correctly captures supported models.

\begin{theorem}[Transformation Correctness]\label{thm:transformation}
Let $\Pi$ be a ground logic program. An interpretation $M$ is a supported model of $\Pi$ if and only if 
\[
M' = M \cup \{\_dm\_r \mid r \in \Pi, \textnormal{ body of } r \textnormal{ is false in } M\}
\]
is a stable model of $T(\Pi)$.
\end{theorem}

\begin{proof}
We prove both directions.

\medskip
\noindent\textbf{($\Rightarrow$) Supported model to stable model:}

Let $M$ be a supported model of $\Pi$. Define:
\[
M' = M \cup \{\_dm\_r \mid r \in \Pi, M \not\models \textnormal{body}(r)\}
\]

We show $M'$ is a stable model of $T(\Pi)$ by proving it equals the least model of the reduct $(T(\Pi))^{M'}$.

\emph{Step 1: $M'$ is a model of $T(\Pi)$.}

Consider any rule in $T(\Pi)$:

\textbf{Case 1:} Rule $h \leftarrow \textnormal{not}\; \_dm\_r$ from some $r: h \leftarrow \textnormal{body}(r)$ in $\Pi$.

If $\_dm\_r \in M'$, the rule is satisfied. If $\_dm\_r \notin M'$, then by construction, $M \models \textnormal{body}(r)$. Since $M$ is a model of $\Pi$ and $M \models \textnormal{body}(r)$, we have $h \in M \subseteq M'$, so the rule is satisfied.

\textbf{Case 2:} Rule $\_dm\_r \leftarrow \textnormal{not}\; l_i$ for $l_i \in \textnormal{body}^+(r)$.

If $l_i \in M'$, the rule is satisfied. If $l_i \notin M'$, then $l_i \notin M$, so $M \not\models \textnormal{body}(r)$, hence $\_dm\_r \in M'$ by construction, so the rule is satisfied.

\textbf{Case 3:} Rule $\_dm\_r \leftarrow l_j$ for $l_j \in \textnormal{body}^-(r)$.

If $l_j \notin M'$, the rule is satisfied. If $l_j \in M'$, then $l_j \in M$, so $M \not\models \textnormal{body}(r)$, hence $\_dm\_r \in M'$ by construction, so the rule is satisfied.

Thus $M'$ is a model of $T(\Pi)$.

\emph{Step 2: $M'$ is the minimal model of $(T(\Pi))^{M'}$.}

The reduct $(T(\Pi))^{M'}$ contains:
\begin{itemize}
\item For each $r$ with $M \models \textnormal{body}(r)$: $\_dm\_r \notin M'$, so rule $h \leftarrow \textnormal{not}\; \_dm\_r$ becomes fact $h \leftarrow$ in the reduct.

\item For each $r$ with $M \not\models \textnormal{body}(r)$: $\_dm\_r \in M'$, so rule $h \leftarrow \textnormal{not}\; \_dm\_r$ is removed from the reduct.

\item For $\_dm\_r$ rules: If $l_i \notin M'$ and $\_dm\_r \in M'$, rule $\_dm\_r \leftarrow \textnormal{not}\; l_i$ becomes fact $\_dm\_r \leftarrow$ in reduct. Similarly for rules $\_dm\_r \leftarrow l_j$ when $l_j \in M'$ and $\_dm\_r \in M'$.
\end{itemize}

Let $N$ be the least model of $(T(\Pi))^{M'}$.

For any original atom $a \in \textnormal{atoms}(\Pi)$: $a \in N$ iff there exists rule $r$ with $\textnormal{head}(r) = a$ and $M \models \textnormal{body}(r)$. Since $M$ is a supported model, every $a \in M$ has such a supporting rule. Conversely, if $a \notin M$ but $a \in N$, then there exists $r$ with $M \models \textnormal{body}(r)$ and $\textnormal{head}(r) = a$, but since $M$ is a model, this implies $a \in M$, contradiction. Thus $N \cap \textnormal{atoms}(\Pi) = M$.

For any auxiliary atom $\_dm\_r$: $\_dm\_r \in N$ iff $M \not\models \textnormal{body}(r)$, which is exactly when $\_dm\_r \in M'$ by construction.

Therefore $N = M'$, so $M'$ is a stable model of $T(\Pi)$.

\medskip
\noindent\textbf{($\Leftarrow$) Stable model to supported model:}

Let $M'$ be a stable model of $T(\Pi)$. Define $M = M' \cap \textnormal{atoms}(\Pi)$.

\emph{Step 1: $M$ is a model of $\Pi$.}

Consider any rule $r: h \leftarrow \textnormal{body}(r)$ in $\Pi$. The transformed program contains $h \leftarrow \textnormal{not}\; \_dm\_r$.

Since $M'$ is a model of $T(\Pi)$, either $\_dm\_r \in M'$ or $h \in M'$.

If $\_dm\_r \in M'$, then the rules generating $\_dm\_r$ imply that $M \not\models \textnormal{body}(r)$, so $r$ is satisfied in $M$.

If $\_dm\_r \notin M'$, then $h \in M'$, hence $h \in M$. Moreover, $\_dm\_r \notin M'$ means no rule generating $\_dm\_r$ has its body satisfied, which means $M \models \textnormal{body}(r)$. Thus $r$ is satisfied in $M$.

Therefore $M$ is a model of $\Pi$.

\emph{Step 2: Every atom in $M$ has support.}

Let $a \in M \subseteq M'$. Since $M'$ is a stable model, it is a supported model of $T(\Pi)$ (by Theorem~\ref{thm:stable-supported}). Thus $a$ has a supporting rule in $T(\Pi)$.

The only rules in $T(\Pi)$ with head $a$ are of the form $a \leftarrow \textnormal{not}\; \_dm\_r$ for some rule $r$ with $\textnormal{head}(r) = a$ in $\Pi$.

For such a rule to support $a$ in $M'$, we need $\_dm\_r \notin M'$. This means $M \models \textnormal{body}(r)$.

Therefore $a$ has a supporting rule $r$ in $\Pi$, so $M$ is a supported model of $\Pi$.
\end{proof}

\subsection{Complexity}

\begin{proposition}\label{prop:complexity}
For a ground program $\Pi$ with $n$ rules and $m$ literals total across all rule bodies, the transformed program $T(\Pi)$ has $O(n + m)$ rules and $O(n + |\textnormal{atoms}(\Pi)|)$ atoms.
\end{proposition}

\begin{proof}
For each rule $r$ in $\Pi$ with $k_r$ literals in its body, we generate 1 rule of form $\textnormal{head}(r) \leftarrow \textnormal{not}\; \_dm\_r$ plus $k_r$ rules for $\_dm\_r$.

Total rules: $n + \sum_{r \in \Pi} k_r = n + m = O(n + m)$.

Total atoms: original atoms plus $n$ new auxiliary atoms = $O(n + |\textnormal{atoms}(\Pi)|)$.
\end{proof}

\section{Implementation and Algorithm Details}

We have implemented our transformation as a preprocessor for Clingo, the state-of-the-art ASP solver.

\textbf{Note on Implementation Scope:} The current implementation focuses on transforming and solving \textit{grounded} ASP programs. The rationale for this restriction is discussed in Section~\ref{sec:grounding}.

\subsection{Algorithm}

\begin{algorithm}
\caption{Computing Supported Models}
\label{alg:supported}
\begin{algorithmic}[1]
\State \textbf{Input:} Ground logic program $\Pi$
\State \textbf{Output:} Set of supported models $\mathcal{M}$
\State $\mathcal{M} \gets \emptyset$
\State $\Pi' \gets \textsc{Transform}(\Pi)$ \Comment{Apply De Morgan's transformation as in Definition~\ref{def:transformation}}
\State $\mathcal{M}' \gets \textsc{Clingo}(\Pi')$ \Comment{Compute stable models of $\Pi'$ using the Clingo solver}
\For{each stable model $M' \in \mathcal{M}'$}
    \State $M \gets M' \cap \textnormal{atoms}(\Pi)$ \Comment{Project to original vocabulary}
    \If{$M \notin \mathcal{M}$}
        \State $\mathcal{M} \gets \mathcal{M} \cup \{M\}$ \Comment{Collect the unique supported model}
    \EndIf
\EndFor
\State \Return $\mathcal{M}$
\end{algorithmic}
\end{algorithm}

\medskip
\noindent\textbf{Detailed Algorithm Explanation}

Algorithm~\ref{alg:supported} outlines the process implemented in our preprocessor tool. The tool first receives a ground logic program $\Pi$ (Line 1). The central step is the call to $\textsc{Transform}(\Pi)$ (Line 4), which applies Definition~\ref{def:transformation} to generate the transformed program $\Pi'$. This function iterates through every rule $r$ in $\Pi$ and replaces it with one rule $h \leftarrow \textnormal{not}\; \_dm\_r$ (Equation~\ref{eq:t1}) and a set of rules (Equations~\ref{eq:t2}, \ref{eq:t3}) that define the auxiliary atom $\_dm\_r$ based on the failure of $r$'s body.

The resulting program $\Pi'$ is then passed to the standard, optimized $\textsc{Clingo}$ solver (Line 5) to compute its stable models $\mathcal{M}'$. As proven in Theorem~\ref{thm:transformation}, each stable model $M'$ of $\Pi'$ uniquely corresponds to a supported model of $\Pi$. The subsequent steps (Lines 6--10) project the auxiliary atoms out of the models in $\mathcal{M}'$ to obtain the supported models $M$ of the original program $\Pi$ and collect them in the final set $\mathcal{M}$. This projection step is essential for presenting the solution in the original program's vocabulary.

\subsection{Usage Example}

Given an input program:
\begin{verbatim}
p :- q, not r.
q :- p.
\end{verbatim}

Our tool generates the transformed program:
\begin{verbatim}
p :- not _dm_r1.
_dm_r1 :- not q.
_dm_r1 :- r.
q :- not _dm_r2.
_dm_r2 :- not p.
\end{verbatim}

Running Clingo computes stable models which, after projection, give supported models: $\emptyset$ and $\{p,q\}$.

\section{Discussion and Related Concepts}

\subsection{Relationship to Clark's Completion and SAT Solvers}
\label{sec:clarks}

The \textbf{supported models} of a logic program $\Pi$ coincide exactly with the \textbf{classical models} of its \textbf{Clark's Completion} (Comp($\Pi$)). Therefore, any tool that computes the models of Comp($\Pi$) can compute supported models. An established alternative approach is to use a dedicated system to generate the Clark's Completion and then use a highly-optimized SAT solver to find its models.

Our contribution is not in proposing a fundamentally new semantics, but in providing a new and practical method for computing supported models \textbf{within the ASP paradigm}. The transformed program $T(\Pi)$ is still an ASP program solved by a stable model solver.

\begin{itemize}
    \item \textbf{Decoupling vs.\ Minimality Enforcement:} Clark's Completion provides the \emph{model} property. ASP solvers then enforce the \emph{minimality} condition (via unfounded set checks) to find stable models. Our transformation $T(\Pi)$ \emph{decouples} these two concerns. The rules defining $\_dm\_r$ capture the supported condition, and the structure of the final rules $h \leftarrow \textnormal{not}\; \_dm\_r$ ensures that the stable model mechanism of Clingo correctly finds \emph{all} supported models of the original program $\Pi$, without inadvertently reintroducing the minimality restriction.
    \item \textbf{Practical Integration:} The key advantage is practical: our method requires only Clingo and an external preprocessor, enabling researchers already using the highly optimized Clingo ecosystem to access supported models seamlessly without requiring a separate SAT solver toolchain.
\end{itemize}

\subsection{The Grounding Restriction}
\label{sec:grounding}

The restriction to grounded logic programs is \textbf{implementation-based, not theoretical}. The transformation $\textsc{Transform}(\Pi)$ can be applied to any ground program $\Pi$. For non-ground programs $\Pi_{\mathit{ng}}$, the overall computation process is:
\[
\Pi_{\mathit{ng}} \xrightarrow{\textnormal{gringo}} \Pi \xrightarrow{\textsc{Transform}} \Pi' \xrightarrow{\textnormal{clasp}} \mathcal{M}' \xrightarrow{\textnormal{Projection}} \mathcal{M}
\]
The transformation itself could, in principle, be integrated into the grounding step. Our current preprocessor approach simply treats the transformation as an operation on the already-ground program $\Pi$. Future work will investigate integrating the transformation into the ASP system's front-end to handle non-ground programs transparently.

\section{Applications}

We demonstrate scenarios where supported models enable valuable exploratory reasoning beyond stable models.

\subsection{Software Verification: Comprehensive Test Generation}

In software testing, we often want to generate all possible system states consistent with component behaviors, not just minimal configurations.

\begin{example}[Redundant Systems Testing]\label{ex:backup}
Consider a system where components can be activated and interact:
\begin{verbatim}
% Components can be self-activating (representing external triggers)
primary :- primary.
backup :- backup.

% System behavior rules
running :- primary.
running :- backup.
redundant :- primary, backup.
\end{verbatim}

The self-referential rules model components that can be externally triggered or activated.

\textbf{Stable models:}
\begin{itemize}
\item $\emptyset$: No components active
\end{itemize}

\textbf{Additional supported models:}
\begin{itemize}
\item \texttt{\{primary, running\}}: Only primary active
\item \texttt{\{backup, running\}}: Only backup active  
\item \texttt{\{primary, backup, running, redundant\}}: Both active (redundant configuration)
\end{itemize}

The supported models capture all consistent system configurations, including the redundant state where both primary and backup are active. For comprehensive testing, verifying system behavior under all these configurations is valuable.
\end{example}

\subsection{Medical Diagnosis: Hypothesis Generation}

In medical diagnosis, physicians often consider multiple co-occurring conditions, not just the minimal set explaining symptoms.

\begin{example}[Differential Diagnosis]\label{ex:medical}
Consider a diagnostic scenario where conditions can be hypothesized:
\begin{verbatim}
% Conditions can be hypothesized (self-supporting for exploration)
infection :- infection.
inflammation :- inflammation.
allergy :- allergy.

% Symptom derivation
fever :- infection.
fever :- inflammation.
cough :- infection.
cough :- allergy.
fatigue :- infection.
fatigue :- inflammation.
\end{verbatim}

\textbf{Stable models:} Only $\emptyset$ (no conditions assumed).

\textbf{Supported models include:}
\begin{itemize}
\item \texttt{\{infection, fever, cough, fatigue\}}: Infection alone
\item \texttt{\{inflammation, fever, fatigue\}}: Inflammation alone
\item \texttt{\{allergy, cough\}}: Allergy alone
\item \texttt{\{infection, inflammation, fever, cough, fatigue\}}: Co-occurring infection and inflammation
\item \texttt{\{infection, allergy, fever, cough, fatigue\}}: Co-occurring infection and allergy
\item \texttt{\{inflammation, allergy, fever, cough, fatigue\}}: Co-occurring inflammation and allergy
\item \texttt{\{infection, inflammation, allergy, fever, cough, fatigue\}}: All three conditions
\end{itemize}

These represent differential diagnoses where multiple conditions may co-occur. A physician exploring all plausible hypotheses would benefit from this enumeration, rather than assuming only minimal explanations.
\end{example}

\subsection{Planning: Alternative Resource Allocations}

In resource allocation, we may want to explore all valid allocations, including those using more resources than strictly necessary.

\begin{example}[Server Allocation]\label{ex:server}
Consider allocating servers to tasks:
\begin{verbatim}
% Servers can be allocated (self-supporting choice)
server1 :- server1.
server2 :- server2.
server3 :- server3.

% Task completion conditions
taskA_done :- server1.
taskA_done :- server2.
taskB_done :- server2.
taskB_done :- server3.
all_done :- taskA_done, taskB_done.
\end{verbatim}

\textbf{Stable models:} Only $\emptyset$.

\textbf{Supported models include:}
\begin{itemize}
\item \texttt{\{server2, taskA\_done, taskB\_done, all\_done\}}: Minimal allocation
\item \texttt{\{server1, server2, taskA\_done, taskB\_done, all\_done\}}: Redundant for task A
\item \texttt{\{server2, server3, taskA\_done, taskB\_done, all\_done\}}: Redundant for task B
\item \texttt{\{server1, server2, server3, taskA\_done, taskB\_done, all\_done\}}: Maximum redundancy
\end{itemize}

Over-provisioned allocations may be preferred for reliability in production environments.
\end{example}

\subsection{When to Use Supported Models}

Based on our experience, supported models are valuable when:

\begin{itemize}
    \item \textbf{Exploration over optimization}: The goal is to enumerate possibilities, not find optimal solutions
    \item \textbf{Non-minimal scenarios matter}: Real-world situations often involve redundancy or co-occurrence
    \item \textbf{Hypothesis generation}: Generating candidate explanations for further investigation
    \item \textbf{Test coverage}: Ensuring systems are tested under all consistent scenarios
\end{itemize}

Stable models remain preferable when:
\begin{itemize}
    \item Minimality directly corresponds to optimality
    \item Occam's razor applies (prefer simpler explanations)
    \item The problem naturally has unique or minimal solutions
\end{itemize}

\section{Related Work}

\subsection{Supported Models Literature}

Supported models were introduced by Apt, Blair, and Walker \cite{apt1988towards} as an alternative semantics for logic programs. Apt and Bol \cite{apt1994logic} provided a comprehensive survey of logic programming with negation, establishing supported models as a relaxation of stable models.

Marek and Subrahmanian \cite{marek1992relationship} studied the relationship between various semantics, proving that stable models form a subset of supported models.

Previous computational approaches relied on fixpoint-based interpreters. Our work provides a transformation-based method enabling computation via standard stable model solvers.

\subsection{Alternative ASP Semantics}

\textbf{Well-founded semantics} \cite{gelder1991well} computes a unique three-valued model, providing polynomial-time reasoning but losing the flexibility of multiple answer sets.

\textbf{Partial stable models} \cite{przymusinski1991stable} extend stable models to three-valued logic, similar in spirit to supported models in relaxing completeness requirements.

\textbf{Coinductive ASP} \cite{gupta2012coinductive} extends standard ASP to capture non-minimal models, particularly for positive cycles, using greatest fixed-point semantics and requiring specialized solvers.

Our work is complementary to these approaches, focusing specifically on supported models and their practical computation using standard ASP infrastructure.

\subsection{Expressiveness of ASP Extensions}

\textbf{Choice rules} like \texttt{\{p\}} allow non-deterministic inclusion of atoms, but require explicit specification for each atom that should be non-minimal. Our transformation systematically generates all supported models without manual annotation.

\textbf{Disjunctive ASP} can encode alternatives, but requires knowing a priori which atoms should be treated non-minimally. Our approach handles this systematically.

Our transformation provides:
\begin{itemize}
\item Automatic generation of all supported models
\item Preservation of the declarative nature of the original program
\item Theoretical guarantees via Theorem~\ref{thm:transformation}
\end{itemize}

\section{Conclusion}

We presented a transformation-based method for computing supported models using standard ASP solvers. Our approach bridges the gap between theoretical semantics (supported models, defined in 1988) and practical tools (modern ASP solvers optimized for stable models).

The key contributions are:
\begin{enumerate}
\item A polynomial-time transformation with proven correctness
\item A practical implementation for Clingo
\item Demonstration of applications where supported models enable valuable exploratory reasoning
\end{enumerate}

\subsection{Future Work}

Promising directions for future work include:

\begin{itemize}
    \item \textbf{Transparent handling of variables}: Integrating the transformation directly within the ASP system's front-end to handle non-ground ASP programs transparently, eliminating the current implementation-based grounding restriction
    \item \textbf{Empirical evaluation}: Conducting a systematic performance comparison between our ASP-based transformation and the established Clark's Completion/SAT solver approach across diverse benchmark programs
    \item \textbf{Preference-based filtering}: Developing frameworks to rank supported models based on domain-specific criteria
    \item \textbf{Incremental computation}: Exploring algorithms for incrementally generating supported models
    \item \textbf{Hybrid approaches}: Combining supported models for some predicates with stable models for others
    \item \textbf{Integration with co-ASP}: Exploring connections between supported model semantics and coinductive ASP
\end{itemize}

We believe supported models deserve renewed attention as a practical tool for exploratory reasoning in knowledge representation applications.

\bibliographystyle{splncs04}
\bibliography{references}

\end{document}